\documentclass{llncs}
\usepackage{fancyhdr}
\usepackage{amsmath}  
\usepackage{amssymb}  
\usepackage{cite}
\usepackage{float}
\usepackage{booktabs}
\usepackage{array}

\usepackage[T1]{fontenc}
\usepackage{graphicx}
\begin{document}
\title{Efficient Encrypted Computation in Convolutional Spiking Neural Networks with TFHE}

%
%

\author{
\makebox[\textwidth][c]{\normalsize
Longfei Guo\inst{1,3},
Pengbo Li\inst{2,3},
Ting Gao\inst{2,3,4}$^{*}$,
Yonghai Zhong\inst{2,3},
Haojie Fan\inst{2,3},
Jinqiao Duan\inst{5,6}
}
}

\institute{\small
\parbox{\textwidth}{\centering\mbox{School of Cyber Science and Engineering, Huazhong University of Science and Technology, Wuhan, China}}
\and
\parbox{\textwidth}{\centering\mbox{School of Mathematics and Statistics, Huazhong University of Science and Technology, Wuhan, China}}
\and
\parbox{\textwidth}{\centering\mbox{Center for Mathematical Science, Huazhong University of Science and Technology, Wuhan, China}}
\and
\parbox{\textwidth}{\centering\mbox{Steklov-Wuhan Institute for Mathematical Exploration, Huazhong University of Science and Technology, China}}
\and
\parbox{\textwidth}{\centering\mbox{Department of Mathematics and Department of Physics, Great Bay University, Dongguan, China}}
\and
\parbox{\textwidth}{\centering\mbox{Guangdong Provincial Key Laboratory of Mathematical and Neural Dynamical Systems, Dongguan, China}}
}

\maketitle              
\thispagestyle{fancy}
\begin{abstract}

With the rapid advancement of AI technology, we have seen more and more concerns on data privacy, leading to some cutting-edge research on machine learning with encrypted computation. Fully Homomorphic Encryption (FHE) is a crucial technology for privacy-preserving computation, while it struggles with continuous non-polynomial functions, as it operates on discrete integers and supports only addition and multiplication. Spiking Neural Networks (SNNs), which use discrete spike signals, naturally complement FHE's characteristics. In this paper, we introduce FHE-DiCSNN, a framework built on the TFHE scheme, utilizing the discrete nature of SNNs for secure and efficient computations. By leveraging bootstrapping techniques, we successfully implement Leaky Integrate-and-Fire (LIF) neuron models on ciphertexts, allowing SNNs of arbitrary depth. Our framework is adaptable to other spiking neuron models, offering a novel approach to homomorphic evaluation of SNNs. Additionally, we integrate convolutional methods inspired by CNNs to enhance accuracy and reduce the simulation time associated with random encoding. Parallel computation techniques further accelerate bootstrapping operations. Experimental results on the MNIST and FashionMNIST datasets validate the effectiveness of FHE-DiCSNN, with a loss of less than 3\% compared to plaintext, respectively, and computation times of under 1 second per prediction. We also apply the model into real medical image classification problems and analyze the parameter optimization and selection.

\end{abstract}
\section{Introduction}
\subsubsection{Privacy-Preserved AI}

In recent years, privacy preservation has garnered significant attention in the field of machine learning. Fully Homomorphic Encryption (FHE) has become the most suitable tool to facilitate privacy-preserving machine learning (PPML) due to its strong cryptographic security and the ability to perform algebraic operations directly on the ciphertext. The foundation of FHE was established in 2009 when Gentry introduced the first fully homomorphic encryption scheme \cite{gentry2010computing} capable of evaluating arbitrary circuits. His pioneering work not only proposed the FHE scheme but also outlined a method for constructing a comprehensive FHE scheme from a model with limited yet sufficient homomorphic evaluation capacity. Inspired by Gentry's groundbreaking contributions, subsequent second-generation schemes like BGV \cite{brakerski2014leveled} and FV \cite{fan2012somewhat} had been proposed. The evolution of FHE schemes had continued with third-generation models such as FHEW \cite{ducas2015fhew}, TFHE \cite{chillotti2020tfhe}, and Gao \cite{case2019fully}, which provided very efficient bootstrapping operations and make unlimited homomorphic computation operations possible. The CKKS scheme \cite{cheon2019full} had attracted considerable interest as a suitable tool for PPML implementation, given its natural handling of encrypted real numbers.


However, for homomorphic encryption machine learning, existing FHE frameworks only support arithmetic operations such as addition and multiplication, and the widely used activation functions such as ReLu, sigmoid are non-arithmetic functions. Traditional prediction frameworks that use arithmetic functions to approximate non-arithmetic functions are very inefficient. In a distinct study \cite{bourse2018fast}, the authors developed the FHE-DiNN framework, a discrete neural network framework predicated on the TFHE scheme. Unlike traditional neural networks, FHE-DiNN had discretized network weights into integers and utilized the sign function as the activation function. The computation of the sign function had been achieved through bootstrapping on ciphertexts. Each neuron's output had been refreshed with noise, thereby enabling the neural network to extend computations to any depth. Although FHE-DiNN offered high computational speed, it had compromised model prediction accuracy. Given the close resemblance between the sign function and the output of Spiking Neural Network(SNN) neurons, this work provided a compelling basis for investigating efficient homomorphic evaluations of SNNs in the context of PPML.

\subsubsection{CNNs and SNNs}
Convolutional Neural Networks (CNNs) are powerful tools in computer vision, known for their high accuracy and automated feature extraction \cite{dhillon2020convolutional}. CNNs rely on three key concepts—local receptive fields, weight sharing, and spatial subsampling—implemented through convolutional and pooling layers. These components eliminate the need for manual feature extraction and speed up training, making CNNs ideal for visual recognition tasks \cite{voulodimos2018deep}. In recent years, CNNs have been widely used in image classification \cite{simonyan2014very}, Natural Language Processing (NLP) \cite{bradbury2016quasi}, object detection \cite{kim2016pvanet}, and video classification \cite{ballas2015delving}, driving significant progress in deep learning.

Spiking Neural Networks (SNNs), the third generation of neural networks \cite{Maass1996NetworksOS}, operate more like biological systems. Unlike Artificial Neural Networks (ANNs), SNNs process information in both space and time, capturing the temporal dynamics of biological neurons. Neurophysiologists have developed various neuron models for SNNs, such as the Hodgkin-Huxley (H-H) model \cite{Hodgkin1952AQD}, leaky integrate-and-fire (LIF) model \cite{Wu2017SpatioTemporalBF}, Izhikevich model \cite{Izhikevich2003SimpleMO}, and the spike response model (SRM) \cite{jolivetSpikeResponseModel2003}. These models use spikes or "action potentials" to communicate information, closely mimicking brain activity. This time-based processing makes SNNs more effective for handling time-series data compared to traditional neural networks.

Convolution Spiking Neural Networks (CSNNs) combine CNNs' spatial feature extraction with SNNs' temporal processing, enhancing computational efficiency and accuracy. Zhou et al. \cite{zhou2021temporal} developed a sophisticated SNN architecture using the VGG16 model and GoogleNet. Zhang et al. \cite{zhang2021rectified} created a deep CSNN with two convolutional and two hidden layers, using a ReL-PSP-based spiking neuron model and temporal backpropagation. Other CSNNs range from conversions of standard CNNs to architectures that apply backpropagation using rate coding or multi-peak neuron strategies \cite{zhang2020temporal}.

In this paper, we propose the FHE-DiCSNN framework. Based on the efficient TFHE scheme, this framework combines the discrete characteristics of the CSNN model and efficient computation to realize the accurate prediction of ciphertext neural networks. Our contributions are summarized as follows:
\begin{itemize}
    \item Proposed FHE-Fire and FHE-Reset Functions Based on Bootstrapping: We introduced FHE-Fire and FHE-Reset functions to enable LIF neuron computations on ciphertext, providing a novel solution for privacy preservation in SNN.
    \item We enhanced the framework by optimizing encoding methods and bootstrapping computations, ensuring no noise accumulation in deep network extensions while improving computation accuracy and efficiency.
    \item Extensive experiments on public datasets and real-world medical datasets demonstrated the feasibility of efficient ciphertext computations, laying a foundation for practical applications of homomorphic encryption techniques.
\end{itemize}

%

\section{Preliminary Knowledge}

This section provides an overview of bootstrapping in the TFHE scheme and background knowledge on SNNs.

\subsection{ TFHE scheme}
\subsubsection{Notations}
Set $\mathbb{Z}_p = \left\{-\frac{p}{2}+1, \ldots, \frac{p}{2}\right\}$ denote a finite ring defined over the set of integers. The message space for homomorphic encryption is defined within this finite ring $\mathbb{Z}_p$.

Consider $N=2^k$ and the cyclotomic polynomial $X^N+1$, then
$$
R_{q, N} \triangleq R / q R \equiv \mathbb{Z}_q[X] /\left(X^N+1\right) \equiv \mathbb{Z}[X] /\left(X^N+1, q\right).
$$
Similarly, we can define the polynomial ring $R_{p, N}$.

\subsubsection{ FHEW-like Cryptosystem} 
\begin{itemize}
\item \noindent\textbf{LWE (Learning With Errors).}We revisit the encryption form of LWE \cite{regev2009lattices}, which is employed to encrypt a message $m \in \mathbb{Z}_p$ as
$$LWE_s(m)=(\mathbf{a}, b)=\left(\mathbf{a},\langle\mathbf{a}, \mathbf{s}\rangle+e+ \left\lfloor\frac{q}{p} m\right\rceil\right) \bmod q, $$
where $\mathbf{a} \in \mathbb{Z}_q^n$, $b \in \mathbb{Z}_q$, and the keys are vectors $\mathbf{s} \in \mathbb{Z}_q^n$.  


\item \noindent\textbf{RLWE (Ring Learning With Errors)\cite{lyubashevsky2010ideal}.} An RLWE ciphertext of a message $m(X) \in R_{p,N}$ can be obtained as follows:
$$R L W E_s(m(X)) =  \left(a(X), b(X)\right), \text{where } b(X) = a(X) \cdot s(X)+e(X)+ \left\lfloor \frac{q}{p} m(X) \right\rceil,$$
 where $a(X) \leftarrow R_{q,N}$ is uniformly chosen at random, and $e(X) \leftarrow \chi_\sigma^n$ is selected from a discrete Gaussian distribution with parameter $\sigma$. 

\item \noindent\textbf{GSW\cite{gentry2013homomorphic}.}
Given a plaintext $m \in \mathbb{Z}_p$, the plaintext $m$ is embedded into a power of a polynomial to obtain $X^m \in R_{p,N}$, which is then encrypted as $RGSW(X^m)$. RGSW enables efficient computation of homomorphic multiplication, denoted as $\diamond$, while effectively controlling noise growth:
$$
\begin{aligned}
& R G S W\left(X^{m_0}\right) \diamond R G S W\left(X^{m_1}\right)=R G S W\left(X^{m_0+m_1}\right),\\
& R L W E\left(X^{m_0}\right) \diamond R G S W\left(X^{m_1}\right)=R L W E\left(X^{m_0+m_1}\right).
\end{aligned}
$$
\end{itemize}

\subsubsection{Programmable Bootstrapping}
TFHE/FHEW bootstrapping is a core algorithm in FHE that enables the computation of any function $g$ with the properties $g: \mathbb{Z}_p \rightarrow \mathbb{Z}_p$ and $g(v + \frac{p}{2}) = -g(v)$. This function $g$, often referred to as the program function in bootstrapping, defines the transformation applied during the bootstrapping process.

Given an LWE ciphertext $LWE(m)_s = (\mathbf{a}, b)$, where $m \in \mathbb{Z}_p$, $\mathbf{a} \in \mathbb{Z}_p^N$, and $b \in \mathbb{Z}_p$, it is possible to bootstrap it into $LWE_s(g(m))$ with refreshed ciphertext noise. This process effectively reduces noise accumulation, ensuring the ciphertext remains suitable for further homomorphic computations.




\begin{equation*}
    R L W E_z \xrightarrow{\text{ Sample Extraction}} L W E_z(g(m))=\left(\mathbf{a}, b_0\right),
    \label{Sample_Extraction}
\end{equation*}
where $\mathbf{a}=\left(a_0, \ldots, a_{N-1}\right)$ is the coefficient vector of $a(X)$, and $b_0$ is the coefficient of $b(X)$.



Given a program function $g$, bootstrapping is a process that takes an LWE ciphertext $LWE_s(m)$ as input and outputs $LWE_s(g(m))$ with the original secret key $s$:
\begin{equation*}
    \text { bootstrapping }\left(L W E_s(m)\right)=L W E_s(g(m)) \text {. }
\end{equation*}

\subsection{Leaky Integrate-and-Fire Neuron Model and SNNs}

Neurophysiologists have developed several models to describe the dynamic behavior of neuronal membrane potentials, crucial for understanding the fundamental properties of SNNs. Among these, LIF model \cite{wu2018spatio} is widely used due to its simplicity and effectiveness in capturing essential neuronal dynamics, such as leakage, accumulation, and threshold excitation. The LIF model is expressed as:
\begin{equation}
    \tau \frac{\mathrm{d} V}{\mathrm{~d} t} = V_{\text{rest}} - V + RI,
\end{equation}
where $\tau$ represents the membrane time constant, $V_{\text{rest}}$ is the resting potential, $R$ is the membrane impedance, and $I$ is the input current. This model can be discretized as:
\begin{equation}
\left\{
\begin{aligned}
    &H[t] = V[t-1] + \frac{1}{\tau}(-(V[t-1] - V_{\text{reset}}) + I[t]), \\
    &S[t] = \text{Fire}(H[t] - V_{\text{th}}), \\
    &V[t] = \operatorname{Reset}(H[t]) = \begin{cases}
       V_{\text{reset}}, & \text{if } H[t] \geq V_{\text{th}}, \\
       H[t], & \text{if } V_{\text{reset}} \leq H[t] \leq V_{\text{th}}, \\
       V_{\text{reset}}, & \text{if } H[t] \leq V_{\text{reset}}.
    \end{cases}
\end{aligned}\right.
\label{LIFnode}
\end{equation}
where  $H[t]$ represents the membrane potential, $S[t]$ the spike function, and $V_{\text{th}}$ the firing threshold. The $\text{Fire}(\cdot)$ function is a step function. The input current $I[t]$ is the weighted sum of external inputs, such as those from pre-synaptic neurons or image pixels, denoted as $I[t] = \sum_j \omega_{j} x_j[t]$. These inputs can be processed using either convolutional or fully connected layers in a neural network.

Training SNNs is challenging due to the non-differentiability of spikes \cite{zhang2018plasticity}. Traditional backpropagation can't be directly applied, leading to alternative training methods like ANN-to-SNN conversion, unsupervised Spike-Timing-Dependent Plasticity (STDP), and gradient surrogate methods. We adopt a gradient surrogate approach, using continuous functions to approximate the derivative of spikes for backpropagation.

To handle diverse input patterns, different encoding methods are used, such as rate coding, temporal coding, and population coding \cite{georgopoulos1986neuronal}. In this study, Poisson encoding was used to convert input data into spike trains using a Poisson process. For a time interval $\Delta t$ divided into $T$ steps, a random matrix $M_t$ with values in $[0, 255]$ is generated. Each time step $t$, the normalized pixel matrix $X_o$ is compared with $M_t$ to determine spike occurrences. The final spike encoding $X$ is:

\[
X(i,j) =
\begin{cases}
0, & \text{if } X_o(i,j) \leq M_t(i,j), \\
1, & \text{if } X_o(i,j) > M_t(i,j),
\end{cases}
\]

where $i$ and $j$ are pixel coordinates. This encoding aligns with a Poisson distribution, effectively adapting input stimuli for SNNs.

\section{Proposed Methods}
\subsection{Convolutional Spiking Neural Networks}
In SNN neural networks, Poisson coding is a simple stochastic coding technique used to convert continuous signals, such as images, into impulse signals, which are thus suitable as inputs to SNN neural networks. However, Poisson coding itself lacks the ability to extract image features. Due to its stochastic nature, signal coding using Poisson coding requires a large number of impulse samples to encode the complete information, which in turn increases the computational cost and time overhead.

Convolutional Spiking Neural Network(CSNN) is a neural network model that combines CNN and SNN, which replaces the use of Poisson coding. In CSNN, the previously mentioned LIF model or other spiking models are used to simulate the electrical activity of the neurons, forming the spiking activation layer in the network, while the convolutional layer of CSNN is used to extract image features and encode them as spiking signals. By utilising the spatial feature extraction capability of CNNs and the spiking properties of SNNs, CSNNs can benefit from both convolutional operations and discrete spiking.A visualisation of the CSNN is shown in Figure (\ref{CSNN_main2}).


\begin{figure}[htbp]
    \centering
		\centering
		\includegraphics[width=1\textwidth]{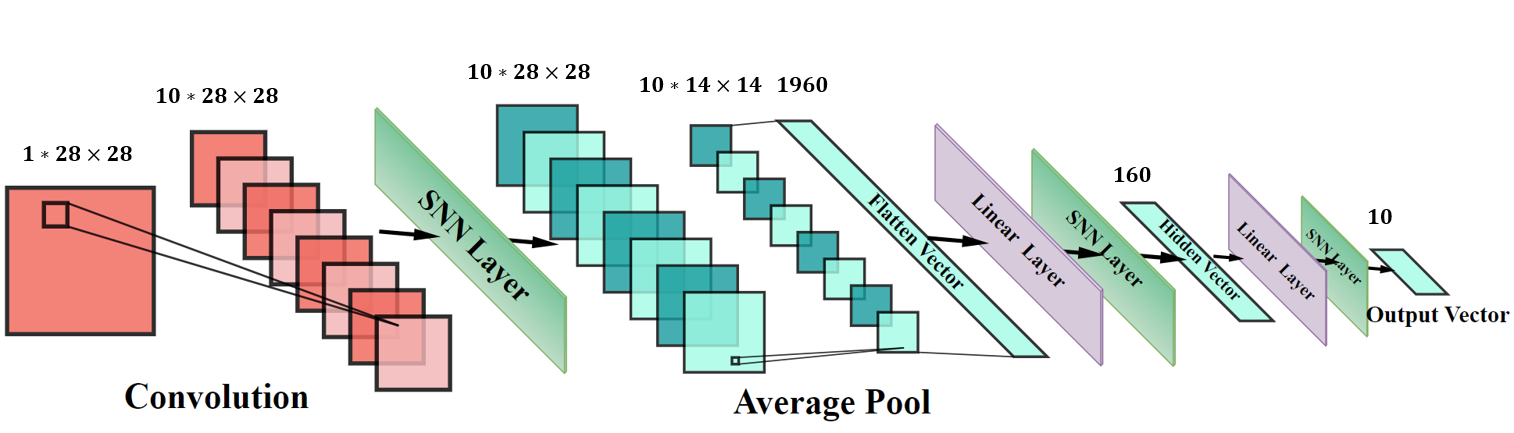}
        \caption{Schematic visualisation of the CSNN network. For the network, we use a convolutional layer with kernal\_size = 3, stridge = 1,padding=1, and an average pooling layer with kernal\_size = 2, stridge = 2, padding=0.}
    \label{CSNN_main2}
\end{figure}

\subsection{Discretized CSNN}
CSNN relies on real numbers for computation, and the outputs of neurons and network weights are represented as continuous values. However, homomorphic encryption algorithms require the computation to be defined over a ring of integers. Therefore, the outputs and weights of the neural network must be discretised into integers. Discretized CSNNs are characterised by the fact that the neuron outputs consist fundamentally of discrete-valued signals and only the weights need to be discretised. From this perspective, CSNNs are more suitable for homomorphic computation than conventional neural networks.

  
We utilize fixed-precision real numbers and apply suitable scaling to convert the weights into integers, effectively discretizing CSNNs into DiCSNNs(Discretized Convolutional Spiking Neural Network). Denote this discretization method as the following function:

$$
\hat{x} \triangleq \operatorname{Discret}(x, \theta)=\lfloor x \cdot \theta\rceil,
$$
where $\theta \in \mathbb{Z}$ is referred to as the scaling factor, and $\lfloor \cdot \rceil$ represents rounding to the nearest integer. The discretized result of $x$ is denoted as $\hat{x}$. Within the encryption process, all relevant numerical values are defined on the finite ring $\mathbb{Z}_p$. Hence, it is imperative to carefully monitor the numerical fluctuations throughout the computation to prevent reductions modulo $p$, as such reductions could give rise to unanticipated errors in the computational outcomes.


Despite scaling and discretizing the weights, several challenges remain in achieving a fully homomorphic encryption-compatible CSNN (FHE-CSNN):  

1) The original equations of LIF neurons (Equation \ref{LIFnode}) and the average pooling layer involve division operations, which are incompatible with homomorphic computation. Alternative approaches are necessary to eliminate explicit division while maintaining FHE compatibility.  

2) The Fire function, representing the step function in LIF neurons, depends on the bootstrap function $g$ in TFHE. This function must satisfy the condition $g(x) = -g\left(\frac{p}{2} + x\right)$. The Sign function meets this requirement and is worth considering as a substitute for the Fire function.  

To address these challenges, we introduce Theorem \ref{discret_LIF}, which explains the iterative process of discrete LIF neurons.

\begin{theorem}\label{discret_LIF}
Under the given conditions $V_{reset} = 0$ and $V_{th} = 1$, Eq.(\ref{LIF_equation}) can be discretized into the following equivalent form with a discretization scaling factor $\theta$:

\begin{equation}
\left\{
\begin{aligned}
& \hat{H}[t]=\hat{V}[t-1]+ \hat{I}[t], \\
& 2S[t]=\text { Sign }\left(\hat{H}[t]-\hat{V}_{\text {th}}^{\tau}\right) +1, \\
& \hat{V}[t]=\operatorname{Reset}(\hat{H}[t])=\left\{\begin{array}{rll}
\hat{V}_{\text {reset }}, & \text { if } \hat{H}[t] \geq \hat{V}_{\text {th}}^{\tau},;\\
\left\lfloor\frac{\tau-1}{\tau}\hat{H}[t]\right\rceil, & \text { if } \hat{V}_{\text {reset }} \leq \hat{H}[t]<\hat{V}_{\text {th }}^{\tau};\\
\hat{V}_{\text {reset }}, & \text { if } \hat{H}[t] \leq \hat{V}_{\text {reset }} .
\end{array}\right.
&
\end{aligned}\right.
\label{DiscretEquation}
\end{equation}

Here, the hat symbol represents the discretized values, and $\hat{V}_{\text {th}}^{\tau} = \theta \cdot \tau$. Moreover $\hat{I}[t]=\sum \hat{\omega}_{j} S_j[t].$
\end{theorem}

\begin{proof}
We multiply both sides of Equations of the LIF model(Eq.(\ref{LIF_equation})) by $\tau$ and the Sign function  has been substituted for the Fire function, which yields the following equations:

\begin{equation*}
\left\{
\begin{aligned}
\tau H[t] & = (\tau - 1)V[t-1] + I[t] ,\\
2S[t] & =\text { Sign }\left(\tau H[t]-V_{\text {th }}^{\tau}\right) + 1, \\
(\tau - 1)V[t] & =\operatorname{Reset}(\tau H[t])=\left\{\begin{array}{rll}
0, & \text { if } & \tau H[t] \geq V_{\text {th}}^{\tau}; \\
\frac{\tau-1}{\tau}(\tau H[t]), & \text { if } & 0 \leq \tau H[t] \leq V_{\text {th}}^{\tau}; \\
0, & \text { if } & \tau H[t] \leq 0 .
\end{array}\right.
\end{aligned}\right.
\end{equation*}
Then, we treat $\tau H[t]$ and $(\tau-1)V[t-1]$ as separate iterations objects. Therefore, we can rewrite $\tau H[t]$ as $H[t]$ without ambiguity, as well as $(\tau-1)V[t-1]$.

Finally, by multiplying the corresponding discretization factor $\theta$, we obtain Eq\ref{DiscretEquation}. Note that since the division operation has been moved to the Reset function, rounding is applied during discretization.

\end{proof}

In Eq.(\ref{DiscretEquation}), the Leaky Integrate-and-Fire (LIF) model degenerates into the Integrate-and-Fire (IF) model when $V_{\text{th}}^{\tau} = 1$ and $\tau = \infty$. To facilitate further discussions, we will refer to this set of equations as the LIF(IF) function.
\begin{equation}\label{LIF_IF}
\begin{aligned}
    2S[t] &= LIF(I[t]),\\
    2S[t] &= IF(I[t]).
\end{aligned}
\end{equation}


\par\par

LIF(IF) model twice spike signals in Eq (\ref{LIF_IF}). If left unaddressed, the next Spiking Activation Layer would receive twice the input. To tackle this issue, we propose a \textbf{multi-level discretization} method that can also resolve the division problem in average pooling.

Suppose the input signal is given by  

\begin{equation}\label{wS}
\begin{aligned}
    I[t] = \sum \omega_j S_j[t], 
\end{aligned}
\end{equation}
 
where $\omega_j $represents the weight parameter.  

For the input to the next layer after the LIF neuron, the normal singular input can be restored by adjusting the scaling factor as follows:  

\begin{equation}\label{dwS}
\begin{aligned}  
\hat{I}[t] & = \sum \operatorname{Discret}(\omega_j, \frac{\theta}{2}) \cdot 2S_j   
\approx \theta \cdot I[t].  
\end{aligned}
\end{equation} 

For the division involved in the pooling layer, the scaling factor is adjusted from $\theta$ to $\theta / n$ as follows:  

  \begin{equation}\label{nwS}
\begin{aligned}  
\hat{I}[t] = \sum_j \operatorname{Discret}(\omega_j, \frac{\theta}{n}) \cdot \left(\sum_k S_k[t]\right) 
\approx \theta \cdot I[t],  
\end{aligned}  
\end{equation}   
where $n$ represents the divisor in the average pooling layer.  

By leveraging Theorem \ref{discret_LIF} and the multilevel discretization method, we can effectively address both the discretization and division challenges, ensuring smooth computation of CSNNs on ciphertexts.





\subsection{Homomorphic Evaluation of DiCSNNs}
In the framework of FHE-DiCSNN, homomorphic encryption computation is reduced to two main processes: the ciphertext computation for WeightSum and the ciphertext computation for the LIF neuron model.

\textbf{Ciphertext Computation for WeightSum.} In this process, user data such as input images are encrypted, while model parameters like weights remain publicly accessible on the server side. Essentially, FHE-WeightSum represents the dot product between the plaintext weight vector and the ciphertext vector of the input layer. This computation is described as:  
$$  
\sum \hat{\omega}_{j} LWE(x_j) = LWE\left(\sum \hat{\omega}_{j} x_j\right).  
$$  

This equation holds under the following conditions: 

1) $\sum_j \hat{\omega}_{j} x_j \in \left[-\frac{p}{2}, \frac{p}{2}\right)$: This is satisfied by choosing a sufficiently large information space $\mathbb{Z}_p$.  

2) The noise remains within permissible bounds: This requires selecting an appropriate initial noise $\sigma$ to ensure that the weighted noise $\sum_j \left|\hat{\omega}_{j} \cdot \sigma\right|$ does not exceed the boundary.  

\textbf{Ciphertext Computation for LIF Neurons.}  
The \textbf{Fire} and \textbf{Reset} functions in Equation (ref{DiscretEquation}) are non-polynomial functions, necessitating the use of programmable bootstrap methods from Theorem (ref{Programtfhe}).  

For the FHE-Fire function, we define $g(m)$ as:  
$$  
g(m) =  
\begin{cases}  
1, & \text{if } m \in \left[0, \frac{p}{2}\right); \\  
-1, & \text{if } m \in \left[-\frac{p}{2}, 0\right).  
\end{cases}  
$$  

The ciphertext computation for FHE-Fire is then given by:  
\begin{equation}
\begin{aligned}  
2 \, \text{FHE-Fire}(\text{LWE}(m)) &= \text{bootstrap}(\text{LWE}(m)) + 1 \\  
&= \begin{cases}  
\text{LWE}(2), & \text{if } m \in \left[0, \frac{p}{2}\right); \\  
\text{LWE}(0), & \text{if } m \in \left[-\frac{p}{2}, 0\right).  
\end{cases} \\  
&= \text{LWE}(\operatorname{Sign}(m) + 1) \\  
&= \text{LWE}(2 \cdot \text{Spike}).  
\end{aligned}    
\end{equation}

For the FHE-Reset function, we define $g(m)$ as:  
$$  
g(m) \triangleq  
\begin{cases}  
0, & \text{if } m \in \left[\hat{V}_{\text{th}}, \frac{p}{2}\right); \\  
\left\lfloor\frac{\theta-1}{\theta} m\right\rceil, & \text{if } m \in \left[0, \hat{V}_{\text{th}}\right); \\  
0, & \text{if } m \in \left[\hat{V}_{\text{th}} - \frac{p}{2}, 0\right); \\  
\frac{p}{2} - \left\lfloor\frac{\theta-1}{\theta} m\right\rceil, & \text{if } m \in \left[-\frac{p}{2}, \hat{V}_{\text{th}} - \frac{p}{2}\right).  
\end{cases}  
$$  

Both functions produce the desired encrypted results. However, if $m$ falls within the interval $\left[-\frac{p}{2}, \hat{V}_{\text{th}} - \frac{p}{2}\right)$, the FHE-Reset function may yield incorrect results. To prevent this, it is essential to ensure that $\hat{H}[t]$ does not enter this range. The following theorem demonstrates that this condition is easily satisfied. 

\begin{theorem}
\label{MAXMIN}    
If $\operatorname{M} \triangleq \hat{V}_{\text {th}}+\max _t(|\hat{I}[t]|)$ and $\operatorname{M} \leq \frac{p}{2}$, then $\hat{H}[t] \in\left[\hat{V}_{\text {threshold }}-\frac{p}{2}, \frac{p}{2}\right)$.
\end{theorem}


\begin{proof}
\begin{align}
\max (\hat{H}[t]) & = \max (\hat{V}[t] + \hat{I}[t]) \leq \frac{\tau-1}{\tau} \hat{V}_{th} + \max_t(|\hat{I}[t]|) < \operatorname{M} \leq \frac{p}{2} \label{maxH} \\
\min (\hat{H}[t]) & = \min (\hat{V}[t] + \hat{I}[t]) \geq -\max_t(|\hat{I}[t]|) \geq -\frac{p}{2} + \hat{V}_{th} \geq -\frac{p}{2} \label{minH}
\end{align}
\end{proof}

The above theorem states that as long as $\operatorname{M} \leq \frac{p}{2}$, the maxima and minima of $\hat{H}[t]$ fall within the interval $\left[\hat{V}_{\text{threshold}} - \frac{p}{2}, \frac{p}{2}\right)$. It can be easily shown that the maximum value generated during CSNN computation must occur in the variable $\hat{H}[t]$.  This result not only validates the correctness of the FHE-Reset function but also facilitates the estimation of the maximum value within the information space. Furthermore, it provides a practical guideline for selecting the parameter $p$ for the information space.  

Additionally, the Bootstrapping procedure refreshes the ciphertext noise, ensuring that subsequent layers of the CSNN inherit the same initial noise level. This capability enables the network to extend to arbitrary depths without concerns about noise accumulation.

\section{Experiments}
In this chapter, we will empirically demonstrate the superior performance of FHE-DiCSNN in terms of time efficiency and accuracy. Firstly, we theoretically analyse the maximum value of FHE-DiCSNN as well as the relationship between the maximum value of the information space and the maximum noise growth with the discretisation factor $\theta$. Secondly, we will experimentally evaluate the actual accuracy and time efficiency of FHE-DiCSNN under different combinations of the attenuation coefficient $\tau$ and the discretisation coefficient $\theta$ with different datasets.

\subsection{Experimental Setup}
\textbf{Datasets.} We use two image datasets in our experiments:

1) MNIST. This dataset includes 70,000 28×28 grayscale images of handwritten digits (0-9). It is split into 60,000 training and 10,000 testing images, with the goal of correctly classifying each handwritten digit.

2) FashionMNIST: This dataset contains 70,000 28×28 grayscale images across 10 clothing categories (e.g., T-shirts, trousers). It is divided into 60,000 training and 10,000 testing images, with the task of accurately identifying the clothing category in each image.

3)Tumor Brain Tumor Classification Dataset: This preprocessed dataset consists of 7,023 28×28 grayscale brain tumor images, classified into four categories: glioma, meningioma, no tumor, and pituitary. It is split into 5,712 training and 1,311 testing images, aiming to determine the presence of a brain tumor based on brain imaging data.\\
\\
\textbf{Baselines.} The performance comparison is made between Plaintext CSNN (Convolutional Spiking Neural Network) as figure \ref{CSNN_main2} showed trained on plaintext data and evaluated on a plaintext test set, and the FHE-DiCSNN (Fully Homomorphic Encryption Distributed Convolutional Spiking Neural Network) evaluated on encrypted data. \\
\\
\textbf{Implementaion Details.} We trained the model using the Adam optimizer with batch size of 256, an initial learning rate of 0.001, and simulations over 4 time-steps. The privacy-preserving framework was implemented in C++ using the TFHE library\footnote{https://github.com/tfhe/tfhe} for homomorphic encryption. During encrypted inference, the LWE parameters were set with ciphertext dimension of 128, TLWE polynomial size of $2^{15}$, and a plaintext message space of \(2^{14}\). We tested various combinations of the attenuation parameter $\tau$ and the discretization scaling factor $\theta$.


\subsection{Experimental Results}

In our experiments, image data were encrypted using LWE ciphertext. The encrypted images were then processed by multiplying them with discretized weights, which were passed through the Spiking Activation Layer. Within this layer, the LIF neuron model was applied to the ciphertext, performing FHE-Fire and FHE-Reset operations. Bootstrapping was optimized using FFT technology and parallel computing. This process (steps 1 to 3) was repeated $T$ times, with the resulting outputs accumulated to generate classification scores. After completing these steps, the ciphertext was decrypted, and the class with the highest score was selected as the final classification result. We evaluated various combinations of the decay parameter $\tau$ and scaling factor $\theta$, with the outcomes detailed in Table \ref{MNISTprediction} and Table \ref{FashionMNISTpredection}.

\begin{table}[h]
    \centering
    \caption{Different $\theta$ and $\tau$ parameters on the MNIST dataset (\%)}
    \setlength{\tabcolsep}{15pt} 
    \begin{tabular}{c c c c c >{\centering\arraybackslash}m{3cm} >{\centering\arraybackslash}m{3cm}}
        \toprule
         & $\tau = 2.0$  & $\tau = 3.0$ & $\tau = 4.0$  & $\tau = \infty$ (IF) \\
        \midrule
        $\theta = 20$  & 97.48 & 97.60 & 97.10 & 97.06  \\
        $\theta = 30$  & 97.66 & 97.50 & 97.60 & 97.23 \\
        $\theta = 40$  & 97.60 & 97.87 & 98.00 & 97.32 \\
        $\theta = 50$  & 97.69 & 97.90 & 97.30 & 97.15 \\
        $\theta = 60$  & 97.72 & 97.70 & 98.00 & 97.25 \\
        $\theta = 100$ & 97.68 & 97.55 & 96.90 & 96.38 \\
        Plaintext      & 98.78 & 98.88 & 98.74 & 98.75 \\
        \bottomrule
    \end{tabular}
    \label{MNISTprediction}
\end{table}

\begin{table}[h]
    \centering
    \caption{Different $\theta$ and $\tau$ parameters on the FashionMNIST dataset (\%)}
    \setlength{\tabcolsep}{15pt} 
    \begin{tabular}{c c c c c}
        \toprule
        & $\tau = 2.0$  & $\tau = 3.0$ & $\tau = 4.0$  & $\tau = \infty$ (IF) \\
        \midrule
        $\theta = 20$  & 86.38 & 85.24 & 86.47 & 82.38 \\
        $\theta = 30$  & 86.65 & 86.37 & 86.97 & 83.20 \\
        $\theta = 40$  & 87.20 & 86.85 & 85.82 & 83.38 \\
        $\theta = 50$  & 87.04 & 86.64 & 79.52 & 83.27 \\
        $\theta = 60$  & 87.29 & 81.69 & 77.21 & 82.67 \\
        $\theta = 100$ & 83.11 & 74.72 & 74.45 & 81.76 \\
        Plaintext      & 90.35 & 90.05 & 89.99 & 90.63 \\
        \bottomrule
    \end{tabular}
    \label{FashionMNISTpredection}
\end{table}

\begin{table}[h]
    \centering
    \caption{Different $\theta$ and $\tau$ parameters on the Tumor dataset (\%)}
    \setlength{\tabcolsep}{15pt} 
    \begin{tabular}{c c c c c}
        \toprule
        & $\tau = 2.0$  & $\tau = 3.0$ & $\tau = 4.0$ \\
        \midrule
        $\theta = 20$  & 81.46 & 87.95 & 83.75 \\
        $\theta = 30$  & 80.42 & 87.87 & 86.50 \\
        $\theta = 40$  & 80.78 & 87.19 & 86.04 \\
        $\theta = 50$  & 80.01 & 87.41 & 85.97 \\
        $\theta = 60$  & 81.79 & 87.53 & 85.66 \\
        $\theta = 100$ & 81.62 & 84.75 & 83.60 \\
        Plaintext      & 93.04 & 93.35 & 93.89 \\
        \bottomrule
    \end{tabular}
    \label{tumor}
\end{table}

The experimental results indicate that both the scaling factor $\theta$ and the decay coefficient $\tau$ have a significant impact on the accuracy of ciphertext predictions. Specifically, $\theta$ has a positive effect on accuracy, as larger $\theta$ values tend to improve the prediction precision within the network. However, it is crucial to ensure that the choice of $\theta$ matches the size of the information space. Otherwise, an excessively large $\theta$ could lead to data overflow, negatively affecting accuracy.

As shown in Table \ref{MNISTprediction} and Table \ref{FashionMNISTpredection}, FHE-DiCSNN achieves maximum accuracies of 98.00\% and 87.29\% on the MNIST and FashionMNIST datasets, respectively. In comparison to plaintext predictions, which reached 98.74\% and 90.35\%, our accuracy losses were less than 1\% and 3\% for the two datasets. This demonstrates that the FHE-DiCSNN model can handle ciphertext data effectively while maintaining high prediction accuracy, validating its reliability.

In addition to the public datasets mentioned above, we also tested our framework on a medical brain tumor dataset. In the medical field, patient data often carries privacy concerns, making it crucial to perform preliminary diagnostics without exposing sensitive patient information. As shown in Table \ref{tumor}, our FHE-DiCSNN framework achieved an accuracy of up to 87.95\% on encrypted Tumor data, with only a 5.4\% reduction compared to plaintext results.

\subsection{Time Consumption}
The simulation time $T$, which can be interpreted as the number of cycles required to process a single image, plays a crucial role in determining the time efficiency of the network. A larger $T$ increases the total number of cycles and consequently the number of bootstrapping operations. In Poisson coding, the inherent randomness demands a sufficiently large $T$ to achieve stable experimental results, which significantly raises the time consumption. In contrast, using a convolutional layer-based spike coding approach enables more stable feature extraction, reducing the simulation time to just 4 cycles. This ensures that Leaky Integrate-and-Fire (LIF) neurons accumulate enough membrane potential to generate spikes effectively.

Bootstrapping, known to be the most time-consuming step in Fully Homomorphic Encryption (FHE), primarily takes place during the computation of LIF neurons in FHE-DiCSNN. Therefore, the number of bootstrapping operations can also serve as a straightforward estimate for the total time consumption. The table below provides a comparison between the CSNN and an equivalent dimension Poisson-coded SNN, as defined in Fig. \ref{CSNN_main2}.

\begin{table}[h]
    \centering
    \caption{$T_1$ and $T_2$ represent the simulation time required for Poisson-encoded SNN and CSNN, respectively, to achieve their respective peak accuracy performances. Typically, $T_1$ falls within the range of [20-100], while $T_2$ falls within the range of [2-4].}
    \setlength{\tabcolsep}{15pt} 
    \huge
    \resizebox{1\textwidth}{!}{
\begin{tabular}{ccc}
   \toprule
     & Poisson-encoded SNN\cite{10462898} & CSNN \\
    \midrule
   bootstrapping & $(784 \times 2 +2\times160+2\times10)\times T_1$ & $(7840\times 2+2\times 160+2\times 10)\times T_2$  \\
   Spiking Activation Layer & $5\times T_1$ & $6\times T_2$  \\
   \bottomrule
\end{tabular}}
\label{Poisson}
\end{table}

\vspace{-15mm}

\begin{table}[h]
    \centering
    \caption{Comparison of the run times of the two models}
    \setlength{\tabcolsep}{15pt} 
    \begin{tabular}{c c c}
        \toprule
        & Multi-threading & Without multi-threading \\
        \midrule
        DiSNN\cite{10462898} & 2.62s & 8.32s \\
        Ours  & 0.34s & 8.29s \\
        \bottomrule
    \end{tabular}
    \label{runtime}
\end{table}

If we do not consider parallel computing, the number of bootstrapping can be used as a simple estimate of the network's time consumption. In this case, the time consumption of both Poisson coded SNNs and CSNNs will be very large. However, since the bootstrapping of Spiking Activation Layers and Poisson encoding can be performed in parallel, the time consumption will be proportional to the number of corresponding layers. In the case of parallel computing, CSNN exhibits a time efficiency that is 10 times higher than that of Poisson-encoded SNN. Table \ref{runtime} presents a comparison of the run times between the two models.

\subsection{Parameters Analysis}
This section discusses the selection of FHE parameters, focusing on the message space $\mathbb{Z}_p$. Here, $p$ serves as the modulus, ensuring operations remain within $\mathbb{Z}_p$. To prevent overflow during subtraction and maintain correctness, Theorem~\ref{MAXMIN} offers a criterion for the maximum allowable value.

As long as it is satisfied that
\begin{equation}
    \hat{V}_{\text{th}} + \max\limits_{t}(|\hat{I}[t]|) \approx \theta \cdot (V_{th}+\max_t (|I[t]|)) = \theta \cdot (V_{th} + \sum\omega_{j}S_j[t]) \leq \frac{p}{2}
\label{select_p}
\end{equation}
the value of the intermediate variable will not exceed the message space $\mathbb{Z}_p$. The formula indicates that $\hat{I}$ is proportional to discretization parameter $\theta$.

We estimated the true maximum value of $V_{th} + \sum\omega_{j}S_j[t]$ on the MNIST dataset, and the findings are summarized in Table \ref{SIZE_P}.



A technique from DiNN \cite{bourse2018fast} dynamically adjusts the message space to minimize computational cost and control noise growth, which we also adopt by selecting a smaller plaintext space. Accurate estimation of noise growth is crucial, as it only occurs during WeightSum operations. For a single WeightSum, the noise increases from the initial value $\sigma$ to:

\begin{equation}
\sigma \sum_j \left| \hat{w}_{j}\right| \approx \theta \cdot \sigma \sum_j \left| w_{j}\right|.
\label{noise_estimate_computing}
\end{equation}

Eq. (\ref{noise_estimate_computing}) shows that the maximum noise growth can be precomputed since weights are known. Experimental results are summarized in Table~\ref{NOISEMAX}.

\begin{table}[h]
\caption{Spiking Activation Layer Results for Different Thresholds $\tau$}
\label{SIZE_P}
\centering
\setlength{\tabcolsep}{15pt} 

\begin{tabular}{lccc}
\toprule
$V_{th} + \max_t (|I[t]|)$ &  Spiking Layer1  &  Spiking Layer2  & Spiking Layer3  \\
\midrule
$\tau=2$ & 29.60 & 70.00 & 24.90 \\
$\tau=3$ & 36.00 & 86.00 & 32.10 \\
$\tau=4$ & 85.2  & 175.00  & 51.70 \\
$\tau=\infty$ (IF) & 29.80 & 31.10 & 20.60 \\
\bottomrule
\end{tabular}
\end{table}

\vspace{-10mm}

\begin{table}[h]
\caption{Based on the noise estimation mentioned above, the quantity $\max\sum_j \left|\omega_{j}\right|$ can be employed to estimate the growth of noise in DiCSNNs for various $\theta$.}
\setlength{\tabcolsep}{15pt} 
\label{NOISEMAX}
\centering
\begin{tabular}{lcccc}
\toprule
 & $\tau=2$ & $\tau=3$ & $\tau=4$ & $\tau=\infty$ (IF) \\
\midrule
$\max\sum_j | \omega_j|$ & 17.42 & 18.87 & 10.24 & 11.89 \\
\bottomrule
\end{tabular}
\end{table}

The discussion shows that the message space size and noise upper bound are proportional to $\theta$. Experimental results provide scaling factors to estimate bounds for different $\theta$, allowing selection of suitable FHE parameters or use of standard sets like STD128 \cite{micciancio2021bootstrapping}.

\section{Conclusion}
This paper introduces the FHE-DiCSNN framework, which is built upon the efficient TFHE scheme and incorporates convolutional operations from CNN. The framework leverages the discrete nature of SNNs to achieve exceptional prediction accuracy and time efficiency in the ciphertext domain. The homomorphic computation of LIF neurons can be extended to other SNNs models, offering a novel solution for privacy protection in third-generation neural networks. Furthermore, by replacing Poisson encoding with convolutional methods, it improves accuracy and mitigates the issue of excessive simulation time caused by randomness. Parallelizing the bootstrapping computation through engineering techniques significantly enhances computational efficiency. Additionally, we provide upper bounds on the maximum value of homomorphic encryption and the growth of noise, supported by experimental results and theoretical analysis, which guide the selection of suitable homomorphic encryption parameters and validate the advantages of the FHE-DiCSNN framework.

There are also promising avenues for future research: 1. Exploring homomorphic computation of non-linear spiking neuron models, such as QIF and EIF. 2. Investigating alternative encoding methods to completely alleviate simulation time concerns for SNNs. 3. Exploring intriguing extensions, such as combining SNNs with RNNs or reinforcement learning and homomorphically evaluating these AI algorithms.

\section{Acknowledgements}
This work was supported by the National Natural Science Foundation of China (12401233), NSFC International Creative Research Team (W2541005), National Key Research and Development Program of China (2021ZD0201300), Guangdong-Dongguan Joint Research Fund (2023A1515140016), Interdisciplinary Research Program of HUST (2023JCYJ012), Guangdong Provincial Key Laboratory of Mathematical and Neural Dynamical Systems (2024B1212010004), Guangdong Major Project of Basic Research (2025B0303000003), and Hubei Key Laboratory of Engineering Modeling and Scientific Computing.

%

%
%
%
\bibliographystyle{splncs04}
%

\bibliography{ref}

\end{document}